\newif \iftheory 

\theorytrue \theoryfalse

\newif \ifblind

\blindtrue \blindfalse

\newcommand{\authorlist}{
Megumi Ando\iftheory\thanks{MITRE, {\tt mando@mitre.org}}\fi
\and 
Anna Lysyanskaya\iftheory\thanks{Computer Science Department, Brown University, {\tt anna@cs.brown.edu}}\fi
\and 
Eli Upfal\iftheory\thanks{Computer Science Department, Brown University, {\tt eli@cs.brown.edu}}\fi
}

\newcommand{\institutelist}{Computer Science Department, Brown University\\Providence, RI 02912 USA\\\tt{megumi\_ando@alumni.brown.edu}, \tt{\{anna, eli\}@cs.brown.edu}}

\newcommand{\titlelist}{\iftheory\begin{bf}\fi
Practical and Provably Secure Onion Routing
\iftheory\end{bf}\fi}

\newcommand{\pathstyles}{stylefiles}

\newcommand{\myparagraph}[1]{\paragraph{#1.}}

\newcommand{\depth}{L}
\newcommand{\blowup}{\Gamma}
\newcommand{\corruptions}{\kappa}

\newcommand{\nprotocol}{\Pi_n}
\newcommand{\pprotocol}{\Pi_p}
\newcommand{\aprotocol}{\Pi_a}

\iftheory
\documentclass[11pt]{article}

\else
\documentclass[envcountsame, 10pt]{\pathstyles/llncs}
\ifblind 
\institute{}
\else
\institute{\institutelist}
\fi
\fi

\ifblind 
\author{}
\else
\author{\authorlist} 
\fi
\title{\titlelist}

\usepackage[utf8]{inputenc}
\usepackage{amsfonts}
\iftheory
\usepackage{amsmath}
\fi
\usepackage{amssymb}
\usepackage{cite}
\usepackage[
	n,
	operators,
	advantage,
	sets,
	adversary,
	landau,
	probability,
	notions,	
	logic,
	ff,
	mm,
	primitives,
	events,
	complexity,
	asymptotics,
	keys]{cryptocode}
\usepackage{enumitem}
\setlist{nosep}
\iftheory
\usepackage[margin=1in]{geometry}
\fi
\usepackage[hidelinks]{hyperref}
\usepackage{todonotes}

\usepackage{lscape}

\iftheory
\usepackage{amsthm}

\newtheorem{lemma}{Lemma}
\newtheorem{theorem}{Theorem}

\newtheorem{definition}{Definition}

\newtheoremstyle{TheoremNum}
        {\topsep}{\topsep}    
        {\itshape}                 
        {}                              
        {\bfseries}                
        {.}                             
        { }                             
        {\thmname{#1}\thmnote{ \bfseries #3}}
\else
\fi

\renewcommand{\secpar}{\lambda}
\renewcommand{\secparam}{1^\lambda}

\begin{document}

\maketitle
\iftheory
\thispagestyle{empty}
\fi

\begin{abstract}
In an onion routing protocol, messages travel through several intermediaries before arriving at their destinations; they are wrapped in layers of encryption (hence they are called ``onions'').  The goal is to make it hard to establish who sent the message.  It is a practical and widespread tool for creating anonymous channels.  

\iftheory\else~~~~\fi For the standard adversary models --- network, passive, and active --- we present practical and provably secure onion routing protocols.  Akin to Tor, in our protocols each party independently chooses the routing paths for his onions.  For security parameter $\secpar$, our differentially private solution for the active adversary takes $\bigO{\log^2\secpar}$ rounds and requires every participant to transmit  $\bigO{\log^{4} \secpar}$ onions in every round. 

\iftheory
\vfill
\noindent\textbf{Keywords:} 
Keyword, 
keyword, 
keyword.
\fi
\end{abstract}

\iftheory
\newpage
\thispagestyle{empty}
\tableofcontents
\newpage
\setcounter{page}{1}
\fi

\section{Introduction}
Anonymous channels are a prerequisite for protecting user privacy.  But how do we achieve anonymous channels in an Internet-like network that consists of point-to-point links?  

If a user Alice wishes to send a message $m$ to a user Bob, she may begin by encrypting her message $m$ under Bob's public key to obtain the ciphertext $c_{\mathit{Bob}} = \enc(\pk_{\mathit{Bob}},m)$.  But sending $c_{\mathit{Bob}}$ directly to Bob would allow an eavesdropper to observe that Alice is in communication with Bob.  So instead, Alice may designate several intermediate relays, called ``mix-nodes'' (typically chosen at random) and send the ciphertext through them, ``wrapped'' in several layers of encryption so that the ciphertext received by a mix-node cannot be linked to the ciphertext sent out by the mix-node. Each node decrypts each ciphertext it receives (``peels off'' a layer of encryption) and discovers the identity of the next node and the ciphertext to send along. This approach to hiding who is talking to whom is called ``onion routing''~\cite{Chaum81} (sometimes it is also called ``anonymous remailer''~\cite{SP:DanDinMat03}) because the ciphertexts are layered, akin to onions; from now on we will refer to such ciphertexts as ``onions''.  

Onion routing is attractive for several reasons: (1)~simplicity: users and developers understand how it works; the only cryptographic tool it uses is encryption; (2)~fault-tolerance: it can easily tolerate and adapt to the failure of a subset of mix-nodes; (3)~scalability: its performance remains the same even as more and more users and mix-nodes are added to the system.
As a result, onion routing is what people use to obscure their online activities. According to current statistics published by the Tor Project, Inc., Tor is used by millions of users every day to add privacy to their communications~\cite{Tor, DM05}\footnote{Tor stands for ``the onion router'', and even though the underlying mechanics are somewhat different from what we described above (instead of using public-key encryption, participants carry out key exchange so that the rest of the communication can be more efficient), the underlying theory is still the same.}.   

In spite of its attractiveness and widespread use, the security of onion routing is not well-understood.  

The definitional question --- what notion of security do we want to achieve? --- has been studied~\cite{SW06,FC:FeiJohSyv07,FJS,EPRINT:BKMMM14}. The most desirable notion, which we will refer to as ``statistical privacy'', requires that the adversary's view in the protocol be distributed statistically independently of who is trying to send messages to whom\footnote{Technically, since onion routing uses encryption, the adversary's view cannot be statistically independent of the input, but at best computationally independent.  However, as we will see, if we work in an idealized encryption model, such as in Canetti's $\mathcal{F}_{\enc}$-hybrid model~\cite{FOCS:Canetti01}, statistical privacy makes sense.}. Unfortunately, a network adversary observing the traffic flowing out of Alice and flowing into Bob can already make inferences about whether Alice is talking to Bob.  For example, if the adversary knows that Alice is sending a movie to someone, but there isn't enough traffic flowing into Bob's computer to suggest that Bob is receiving a movie, then Bob cannot be Alice's interlocutor. (Participants' inputs may also affect others' privacy in other ways~\cite{FJS}.)

So let us consider the setting in which, in principle, statistical privacy can be achieved: every party wants to anonymously send and receive just one short message to and from some other party. Let us call this ``the simple input-output (I/O) setting''.  In the simple I/O setting, anonymity can be achieved even against an adversary who can observe the entire network if there is a trusted party through whom all messages are routed. Can onion routing that does not rely on one trusted party emulate such a trusted party in the presence of a powerful adversary?  

Specifically, we may be dealing with \textit{the network adversary} that observes all network traffic; or the stronger \textit{passive adversary} that, in addition to observing network traffic, also observes the internal states of a fraction of the network nodes; or the most realistic \textit{active adversary} that observes network traffic and also controls a fraction of the nodes. Prior work analyzing Tor~\cite{FC:FeiJohSyv07,FJS,EPRINT:BKMMM14} did not consider these standard adversary models. Instead, they focused on the adversary who was entirely absent from some regions of the network, but resourceful adversaries (such as the NSA) and adversaries running sophisticated attacks (such as BGP hijacking~\cite{SP:SEFCM17}) may receive the full view of the network traffic, and may also infiltrate the collection of mix-nodes. 

Surprisingly, despite its real-world importance, we were the first to consider this question. 

\myparagraph{Warm-up} An \emph{oblivious permutation algorithm} between a memory-constrained client and an untrusted storage server enables the client to permute a sequence of (encrypted) data blocks stored on the server without the server learning anything (in the statistical sense) about the permutation.  
\begin{theorem}
Any oblivious permutation algorithms can be adapted into a communications protocol for achieving statistical privacy from the network adversary. 
\end{theorem}
As an example, Ohrimenko et al.~\cite{ICALP:OGTU14} presented a family of efficient oblivious permutation algorithms. This can be adapted into a secure and ``tunable'' OR protocol that can trade off between low server load and latency. Letting $\secpar$ denote the security parameter, for any $B \in \left[\frac{\sqrt{N}}{\log^2 \secpar}\right]$, this protocol can be set to run in $\bigO{\frac{\log N}{\log B}}$~rounds with communication complexity overhead~$\bigO{\frac{B\log N \log^2 \secpar}{\log B}}$ and server load~$\bigO{B\log^2 \secpar}$. (See Appendix~\ref{sec:pi1}.)  

\myparagraph{Our result for the passive adversary setting} 
However, to be secure from the passive adversary, we need more resources. We prove for the first time that onion routing can provide statistical privacy from the passive adversary, while being efficient.  
\begin{enumerate}
\item We prove that our solution, $\pprotocol$, is statistically private from any passive adversary capable of monitoring any constant $\corruptions \in [0, 1)$ of the mix-nodes, while having communication complexity overhead~$\bigO{\log^{2} \secpar}$, server load~$\bigO{\log^{2} \secpar}$, and latency~$\bigO{\log^{2} \secpar}$, where $\secpar$ denotes the security parameter. (See Section~\ref{sec:pi2}.) 
\end{enumerate}

\myparagraph{Our result for the active adversary setting} 
However, for most realistic input settings (not constrained to the simple I/O setting), statistical privacy is too ambitious a goal. It is not attainable even with a trusted third party.  Following recent literature~\cite{EPRINT:BKMMM14, vdHLZZ15}, for our final result, let us \textit{not} restrict users' inputs, and settle for a weaker notion of privacy, namely, differential privacy. 

Our definition of differential privacy requires that the difference between the adversary's view when Alice sends a message to Bob and its view when she does not send a message at all or sends it to Carol instead, is small. This is meaningful; showing that the protocol achieves differential privacy gives every user a guarantee that sending her message through does not change the adversary's observations very much. 
 
\begin{enumerate}
\setcounter{enumi}{1}
\item Our solution, $\aprotocol$, can defend against the active adversary while having communication complexity overhead~$\bigO{\log^{6} \secpar}$, server load~$\bigO{\log^{4} \secpar}$, and latency~$\bigO{\log^2 \secpar}$. This is the first provably secure peer-to-peer solution that also provides a level of robustness; unless the adversary forces the honest players to abort the protocol run, most messages that are not dropped by the adversary are delivered to their final destinations. (See Section~\ref{sec:pi3}.) 
\end{enumerate}

To prepare onions, we use a cryptographic scheme that is strong enough that, effectively, the only thing that the active adversary can do with onions generated by honest parties is to drop them (see the onion cryptosystem by Camenisch and Lysyanskaya~\cite{C:CamLys05} for an example of a sufficiently strong cryptosystem).  Unfortunately, even with such a scheme, it is still tricky to protect Alice's privacy against an adversary that targets Alice specifically. Suppose that an adversarial Bob is expecting a message of a particular form from an anonymous interlocutor, and wants to figure out if it was Alice or not.  If the adversary succeeds in blocking all of Alice's onions and not too many of the onions from other parties, and then Bob never receives the expected message, then the adversary's hunch that it was Alice will be confirmed. 

How do we prevent this attack?  For this attack to work, the adversary would have to drop a large number of onions --- there is enough cover traffic in our protocol that dropping just a few onions does not do much.  But once a large enough number of onions is dropped, the honest mix-nodes will detect that an attack is taking place, and will shut down before any onions are delivered to their destinations. Specifically, if enough onions survive half of the rounds, then privacy is guaranteed through having sufficient cover; otherwise, privacy is guaranteed because \emph{no} message reaches its final destination with overwhelming probability.  So the adversary does not learn anything about the destination of Alice's onions.

In order to make it possible for the mix-nodes to detect that an attack is taking place, our honest users create ``checkpoint'' onions.  These onions don't carry any messages; instead, they are designed to be ``verified'' by a particular mix-node in a particular round.  These checkpoint onions are expected by the mix-node, so if one of them does not arrive, the mix-node in question realizes that something is wrong.  If enough checkpoint onions are missing, the mix-node determines that an attack is underway and shuts down.  Two different users, Alice and Allison, use a PRF with a shared key (this shared key need not be pre-computed, but can instead be derived from a discrete-log based public-key infrastructure under the decisional Diffie-Hellman assumption; see Appendix~\ref{sec:reallife}) in order to determine whether Alice should create a checkpoint onion that will mirror Allison's checkpoint onion.

\myparagraph{Related work} 
Encryption schemes that are appropriate for onion routing are known~\cite{C:CamLys05,BGKM}. Several papers attempted to define anonymity for communications protocols and to analyze Tor~\cite{SW06,FC:FeiJohSyv07,FJS}. Backes et al.~\cite{EPRINT:BKMMM14} were the first to consider a notion inspired by differential privacy~\cite{tcc:dmns06} but, in analyzing Tor, they assume an adversary with only a partial view of the network. There are also some studies on anonymity protocols, other than onion routing protocols, that were analyzed using information-theoretic measures~\cite{FC:BerFiaTaS04,CCS:KopBas07,ChaPalPan08,EC:DodReySmi04,AlvAndCha+12}. In contrast, all the protocols presented in this paper are provably secure against powerful adversaries that can observe all network traffic. 
The system, Vuvuzela~\cite{vdHLZZ15}, assumes that all messages travel through the same set of dedicated servers and is, therefore, impractical compared to Tor. Recently proposed systems, Stadium~\cite{TGL+17} and Atom~\cite{KCGDF17} are distributed but not robust; they rely on verifiable shuffling to detect and abort. A variant of Atom is robust at a cost in security; it only achieves $k$-anonymity~\cite{KCGDF17}. In contrast, our solution for the active adversary is distributed while maintaining low latency, and robust while being provably secure. 

Achieving anonymous channels using heavier cryptographic machinery has been considered also.  
One of the earliest examples is Chaum's dining cryptographer's protocol~\cite{JC:Chaum88}. Rackoff and Simon~\cite{STOC:RacSim93} use secure multiparty computation for providing security from active adversaries. 
Other cryptographic tools used in constructing anonymity protocols include oblivious RAM (ORAM) and private information retrieval (PIR)~\cite{SP:CooBir95, SP:CorBonMaz15}. Corrigan-Gibbs et al.'s Riposte solution makes use of a global bulletin board and has a latency of a couple of days~\cite{SP:CorBonMaz15}. The aforementioned Stadium~\cite{KCGDF17} is another solution for a public forum. Blaze et al.~\cite{BIK+} provided an anonymity protocol in the wireless (rather than point-to-point) setting.

\section{Preliminaries}
\myparagraph{Notation}
By the notation~$[n]$, we mean the set~$\{1,\dots, n\}$ of integers. The output~$a$ of an algorithm~$A$ is denoted by $a \gets A$. For a set~$S$, we write $s \sample S$ to represent that $s$ is a uniformly random sample from the set~$S$ and $|S|$, to represent its cardinality. A realization~$d$ of a distribution~$D$ is denoted $d \sim D$; by $d \sim {\sf Binomial}(N, p)$, we mean that $d$ is a realization of a binomial random variable with parameters~$N$ and $p$. By $\log(n)$, we mean the logarithm of $n$, base $2$; and by $\ln(n)$, we mean the natural log of $n$. 

A function $f : \mathbb{N} \rightarrow \mathbb{R}$ is negligible in $\secpar$, written $f(\secpar) = \negl$, if for every polynomial~$p(\cdot)$ and all sufficiently large $\secpar$, $f(\secpar) < 1/p(\secpar)$. When $\secpar$ is the security parameter, we say that an event occurs with overwhelming probability if it is the complement of an event with probability negligible in $\secpar$. Two families of distributions~$\{D_{0, \secpar}\}_{\secpar \in\mathbb{N}}$ and $\{D_{1, \secpar}\}_{\secpar \in\mathbb{N}}$ are statistically close if the statistical distance between $D_{0, \secpar}$ and $D_{1, \secpar}$ is negligible in $\secpar$; we abbreviate this notion by $D_0 \approx_s D_1$ when the security parameter is clear by context. We use the standard notion of a pseudorandom function~\cite[Ch.~3.6]{Goldreich01}.

\myparagraph{Onion routing}
Following Camenisch and Lysyanskaya's work on cryptographic onions~\cite{C:CamLys05}, an \emph{onion routing scheme} is a triple of algorithms: 
\[
(\mathsf{Gen}, \mathsf{FormOnion}, \mathsf{ProcOnion}) .
\]
The algorithm, $\mathsf{Gen}$, generates a public-key infrastructure for a set of parties. The algorithm, $\mathsf{FormOnion}$, forms onions; and the algorithm, $\mathsf{ProcOnion}$, processes onions. 

Given a set $[N]$ of parties, for every $i\in[N]$, let $(\pk_i, \sk_i) \gets \mathsf{Gen}(\secparam)$ be the key pair generated for party~$i\in[N]$, where $\secpar$ denotes the security parameter. 

$\mathsf{FormOnion}$ takes as input: a message~$m$, an ordered list $(P_1, \dots, P_{\depth+1})$ of parties from~$[N]$, and the public-keys~$(\pk_{P_1}, \dots, \pk_{P_{\depth+1}})$ associated with these parties, and a list $(s_1, \dots, s_{\depth})$ of (possibly empty) strings that are nonces associated with layers of the onion. The party~$P_{\depth+1}$ is interpreted as the \emph{recipient} of the message, and the list~$(P_1, \dots, P_{\depth+1})$ is the \emph{routing path} of the message. The output of $\mathsf{FormOnion}$ is a sequence $(O_1, \dots, O_{\depth+1})$ of onions. Because it is convenient to think of an onion as a layered encryption object, where processing an onion $O_r$ produces the next onion $O_{r+1}$, we sometimes refer to the process of revealing the next layer of an onion as ``decrypting the onion'', or ``peeling the onion''. For every $r \in [\depth]$, only party~$P_r$ can peel onion~$O_r$ to reveal the next layer,  
\[
(P_{r+1}, O_{r+1}, s_{r+1}) \gets \mathsf{ProcOnion}(\sk_{P_r}, O_r, P_r) ,
\]
of the onion containing the ``peeled'' onion~$O_{r+1}$, the ``next destination''~$P_{r+1}$, and the nonce~$s_{r+1}$. Only the recipient $P_{\depth+1}$ can peel the innermost onion~$O_{\depth+1}$ to reveal the message,
\[
m \gets \mathsf{ProcOnion}(\sk_{P_{\depth+1}}, O_{\depth+1}, P_{\depth+1}) .
\]

Let $O_0$ be an onion formed from running $\mathsf{FormOnion}(m_0, P^0, \pk^0, s^0)$, and let $O_1$ be another onion formed from running $\mathsf{FormOnion}(m_1, P^1, \pk^1, s^1)$. Importantly, a party that can't peel either onion can't tell which input produced which onion. See Camenisch and Lysyanskaya's paper~\cite{C:CamLys05} for formal definitions.

In our protocols, a sender of a message~$m$ to a recipient~$j$ ``forms an onion'' by generating nonces and running the $\mathsf{FormOnion}$ algorithm on the message~$m$, a routing path~$(P_1, \dots, P_{\depth}, j)$, the public keys $(\pk_{P_1}, \dots, \pk_{P_{\depth}}, \pk_j)$ associated with the parties on the routing path, and the generated nonces; the ``formed onion'' is the first onion~$O_1$ from the list of outputted onions. The sender sends $O_1$ to the first party~$P_1$ on the routing path, who processes it and sends the peeled onion~$O_2$ to the next destination~$P_2$, and so on, until the last onion~$O_{\depth+1}$ is received by the recipient~$j$, who processes it to obtain the message~$m$.

\section{Definitions} \label{sec:defns}
We model the network as a graph with $N$ nodes, and we assume that these nodes are synchronized. This way, any onion can be sent from any sender to any receiver, and also its transmission occurs within a single round. 

Every participant is a user client, and some user clients also serve as mix-nodes. In all the definitions, the $N$ users participating in an communications protocol $\Pi$ are labeled $1, \dots, N$; and the number $N$ of users is assumed to be polynomially-bounded in the security parameter~$\secpar$. Every input to a protocol is an $N$-dimensional vector. When a protocol runs on input $\sigma = (\sigma_1, \dots, \sigma_N)$, it means that the protocol is instantiated with each user~$i$ receiving $\sigma_i$ as input. $\mathcal{M}$ denotes the (bounded) message space. A message pair~$(m, j)$ is properly formed if $m \in \mathcal{M}$ and $j \in [N]$. The input~$\sigma_i$ to each user~$i\in[N]$ is a collection of properly formed message pairs, where $(m, j) \in \sigma_i$ means that user~$i$ intends on sending message~$m$ to user~$j$. Let $M(\sigma)$ denote the ``messages in $\sigma$''. It is the multiset of all message pairs in $\sigma$, that is
\[
M(\sigma_1, \dots, \sigma_N) = \bigcup_{i = 1}^N  \left\{ (m, j) \in \sigma_i \right\} .
\]

For analyzing our solutions, it is helpful to first assume an idealized version of an encryption scheme, in which the ciphertexts are information-theoretically unrelated to the plaintexts that they encrypt and reveal nothing but the length of the plaintext.  Obviously, such encryption schemes do not exist computationally, but only in a hybrid model with an oracle that realizes an ideal encryption functionality, such as that of Canetti~\cite{FOCS:Canetti01}. When used in forming onions, such an encryption scheme gives rise to onions that are information-theoretically independent of their contents, destinations, and identities of the mix-nodes.  Our real-life proposal, of course, will use standard computationally secure encryption~\cite{DolDwoNao00}. We discuss the implications of this in Appendix~\ref{sec:reallife}. 

\myparagraph{Views and outputs}
We consider the following standard adversary models, in increasing order of capabilities:
\begin{enumerate}
\item {\bf Network adversary.} A network adversary can observe the bits flowing on every link of the network. (Note that if the peer-to-peer links are encrypted in an idealized sense, then the only information that the adversary can use is the volume flow.) 

\item {\bf Passive adversary.} In addition to the capabilities of a network adversary, a passive adversary can monitor the internal states and operations of a constant fraction of the parties. The adversary's choices for which parties to monitor are made non-adaptively over the course of the execution run. 

\item {\bf Active adversary.} In addition to the capabilities of a network adversary, an active adversary can corrupt a constant fraction of the parties. The adversary's choices for which parties to corrupt are made non-adaptively over the course of the execution run. The adversary can change the behavior of corrupted parties to deviate arbitrarily from the protocol. 
\end{enumerate}

Let $\Pi$ be a protocol, and let $\sigma$ be a vector of inputs to $\Pi$. Given an adversary~$\adv$, the view~$V^{\Pi, \adv}(\sigma)$ of $\adv$ is its observables from participating in $\Pi$ on input~$\sigma$ plus any randomness used to make its decisions. With idealized secure peer-to-peer links, the observables for a network adversary are the traffic volumes on all links; whereas for the passive and active adversaries, the observables additionally include the internal states and computations of all monitored\thinspace/\thinspace corrupted parties at all times. 

Given an adversary $\adv$, the output $O^{\Pi, \adv}(\sigma) = (O_1^{\Pi, \adv}(\sigma), \dots, O_N^{\Pi, \adv}(\sigma))$ of $\Pi$ on input $\sigma$ is a vector of outputs for the $N$ parties. 

\subsection{Privacy definitions}
How do we define security for an anonymous channel? The adversary's view also includes the internal states of corrupted parties. In such case, we may wish to protect the identities of honest senders from the recipients that are in cahoots with the adversary. However, even an ideal anonymous channel cannot prevent the contents of messages (including the volumes of messages) from providing a clue on who sent the messages; thus any ``message content'' leakage should be outside the purview of an anonymous channel. To that end, we say that an communications protocol is secure if it is difficult for the adversary to learn who is communicating with whom, beyond what leaks from captured messages. 

Below, we provide two flavors of this security notion; we will prove that our constructions achieve either statistical privacy or $(\epsilon, \delta)$-differential privacy~\cite[Defn.~2.4]{Dworkbook} in the idealized encryption setting. 

\begin{definition} [Statistical privacy]
Let $\Sigma^{*}$ be the input set consisting of every input of the form 
\[
\sigma = (\{ (m_1, \pi(1))\}, \dots, \{ (m_N, \pi(N))\}) ,
\]
where $m_1, \dots, m_N \in \mathcal{M}$, and $\pi : [N] \mapsto [N]$ is any permutation function over the set $[N]$. 
A communications protocol~$\Pi$ is \emph{statistically private} from the adversaries in the class~$\mathbb{A}$ if for all $\adv \in \mathbb{A}$ and for all $\sigma_0, \sigma_1 \in \Sigma^{*}$ that differ only on the honest parties' inputs and outputs, the adversary's views $V^{\Pi, \adv}(\sigma_0)$ and $V^{\Pi, \adv}(\sigma_1)$ are statistically indistinguishable, i.e., 
\[
\Delta(V^{\Pi, \adv}(\sigma_0), V^{\Pi, \adv}(\sigma_{1})) = \negl ,
\]
where $\secpar \in \mathbb{N}$ denotes the security parameter, and $\Delta(\cdot, \cdot)$ denotes statistical distance (i.e., total variation distance). $\Pi$ is perfectly secure if the statistical distance is zero instead. 
\end{definition}

\begin{definition} [Distance between inputs]
The \emph{distance} between two inputs $\sigma_0 = (\sigma_{0, 1}, \dots, \sigma_{0, N})$ and $\sigma_1 = (\sigma_{1, 1}, \dots, \sigma_{1, N})$, denoted $d(\sigma_0, \sigma_1)$, is given by
\[
d(\sigma_0, \sigma_1) \stackrel{\text{def}}{=} \sum_{i=1}^N |\sigma_{0, i} \nabla \sigma_{1, i}| ,
\]
where $(\cdot \nabla \cdot)$ denotes the symmetric difference. 
\end{definition}

\begin{definition} [Neighboring inputs] 
Two inputs $\sigma_0$ and $\sigma_1$ are \emph{neighboring} if $d(\sigma_0, \sigma_1) \le 1$. 
\end{definition}
 
\begin{definition} [$(\epsilon, \delta)$-DP~{\cite[Defn.~2.4]{Dworkbook}}] 
Let $\Sigma$ be the set of all valid inputs. A communications protocol is \emph{$(\epsilon, \delta)$-DP} from the adversaries in the class~$\mathbb{A}$ if for all $\adv\in\mathbb{A}$, for every pair of neighboring inputs $\sigma_0, \sigma_1 \in \Sigma$ that differ only on an honest party's input and an honest party's output, and any set $\mathcal{V}$ of views,
\[
\prob{V^{\Pi, \adv}(\sigma_{0}) \in \mathcal{V}} \le e^\epsilon \cdot \prob{V^{\Pi, \adv}(\sigma_{1}) \in \mathcal{V}} + \delta .
\]
\end{definition}
While differential privacy is defined with respect to neighboring inputs, it also provides (albeit weaker) guarantees for non-neighboring inputs; it is known that the security parameters degrade proportionally in the distance between the inputs~\cite{Dworkbook}.

\subsection{Other performance metrics}
Since message delivery cannot be guaranteed in the presence of an active adversary, we define correctness with respect to \emph{passive} adversaries. 
\begin{definition} [Correctness]
A communications protocol $\Pi$ is \emph{correct} on an input~$\sigma\in\Sigma$ if for any passive adversary $\adv$, and for every recipient~$j \in [N]$, the output $O_j^{\Pi, \adv}(\sigma)$ corresponds to the multiset of all messages for recipient~$j$ in the input vector~$\sigma$. That is,  
\[
O_j^{\Pi, \adv}(\sigma) = \left\{ m \left| (m, j) \in M(\sigma) \right. \right\} , 
\]
where $M(\sigma)$ denotes the multiset of all messages in $\sigma$. 
\end{definition}

\myparagraph{Efficiency of OR protocols} 
The communication complexity blow-up of an onion routing (OR) protocol measures how many more onion transmissions are required by the protocol, compared with transmitting the messages in onions directly from the senders to the recipients (without passing through intermediaries). We assume that every message~$m\in\mathcal{M}$ in the message space~$\mathcal{M}$ ``fits'' into a single onion. The communication complexity is measured in unit onions, which is appropriate when the parties pass primarily onions to each other. 
\begin{definition} [Communication complexity blow-up]
The \emph{communication complexity blow-up} of an OR protocol~$\Pi$ is defined with respect to an input vector~$\sigma$ and an adversary~$\adv$. Denoted~$\gamma^{\Pi, \adv}(\sigma)$, it is the expected ratio between the total number~$\blowup^{\Pi, \adv}(\sigma)$ of onions transmitted in protocol~$\Pi$ and the total number~$|M(\sigma)|$ of messages in the input vector. That is, 
\[
\gamma^{\Pi, \adv}(\sigma) \stackrel{\text{def}}{=}  \mathbb{E} \left[ \frac{\blowup^{\Pi, \adv}(\sigma)}{\left| M(\sigma) \right|} \right]  .
\]
\end{definition}

\begin{definition} [Server load]
The \emph{server load} of an OR protocol~$\Pi$ is defined with respect to an input vector~$\sigma$ and an adversary~$\adv$. It is the expected number of onions processed by a single party in a round. 
\end{definition}

\begin{definition} [Latency]
The \emph{latency} of an OR protocol~$\Pi$ is defined with respect to an input vector~$\sigma$ and an adversary~$\adv$. It is the expected number of rounds in a protocol execution. 
\end{definition}

In addition to having low (i.e., polylog in the security parameter) communication complexity blow-ups, we will show that our OR~protocols have low (i.e., polylog in the security parameter) server load and low (i.e., polylog in the security parameter) latency.

\section{The passive adversary} \label{sec:pi2}
Communication patterns can trivially be hidden by sending every message to every participant in the network, but this solution is not scalable as it requires a communication complexity blow-up that is linear in the number of participants. Here, we prove that an OR protocol can provide anonymity from the passive adversary while being practical with low communication complexity and low server load. 

To do this, every user must send and receive the same number of messages as any other user; otherwise, the sender-receiver relation can leak from the differing volumes of messages sent and received by the users. In other words, every user essentially commits to sending a message, be it the empty message $\bot$ to itself. Let $\Sigma^{*}$ be the set of all input vectors of the form 
\[
\sigma = (\{ (m_1, \pi(1))\}, \dots, \{ (m_N, \pi(N))\}),
\]
where $m_1, \dots, m_N$ are any messages from the message space $\mathcal{M}$, and $\pi : [N] \rightarrow [N]$ is any permutation function over the set $[N]$; our solution, $\pprotocol$, is presented in the setting where the input vector is constrained to $\Sigma^{*}$. 

\sloppy Let $[N]$ be the set of users, and $\mathcal{S} = \{S_1, \dots, S_{n}\} \subset [N]$ the set of servers. $\pprotocol$ uses a secure onion routing scheme, denoted by $\mathcal{OR} = (\mathsf{Gen}, \mathsf{FormOnion}, \mathsf{ProcOnion})$, as a primitive building block. For every $i\in[N]$, let $(\pk_i, \sk_i) \gets \mathsf{Gen}(\secparam)$ be the key pair generated for party~$i$, where $\secpar$ denotes the security parameter. 

During a setup phase, each user~$i\in[N]$ creates an onion. On input~$\sigma_i = \{(m, j)\}$, user~$i$ first picks a sequence $T_1, \dots, T_\depth$ servers, where each server is chosen independently and uniformly at random, and then forms an onion from the message~$m$, the routing path~$(T_1, \dots, T_\depth, j)$, the public keys $(\pk_{T_1}, \dots, \pk_{T_\depth}, \pk_j)$ associated with the parties on the routing path, and a list of empty nonces. At the first round of the protocol run, user~$i$ sends the formed onion to the first hop~$T_1$ on the routing path. 

After every round~$i\in[L]$ (but before round~$i+1$) of the protocol run, each server processes the onions it received at round~$i$. At round~$i+1$, the resulting peeled onions are sent to their respective next destinations in random order. At round~$\depth+1$, every user receives an onion and processes it to reveal a message.  

\myparagraph{Correctness and efficiency} Clearly, $\pprotocol$ is correct. In $\pprotocol$, $N$ messages are transmitted in each of the $\depth+1$ rounds of the protocol run. Thus, the communication complexity blow-up and the latency are both $\depth+1$. The server load is~$\frac{N}{n}$. 

\myparagraph{Privacy}
To prove that $\pprotocol$ is statistically private from the passive adversary, we first prove that it is secure from the network adversary.    

\begin{theorem} \label{thm-random} 
$\pprotocol$ is statistically private from the network adversary when $\frac{N}{n} = \bigOmega{\log^{2} \secpar}$, and $\depth = \bigOmega{\log^{2} \secpar}$, where $\secpar \in \mathbb{N}$ denotes the security parameter.  
\end{theorem}

\begin{proof}
Let $U \in [N]$ be any target sender. 

Because the network adversary observes every link of the network, the adversary observes the first hop taken by $U$'s onion at the first round of the protocol execution and knows, with certainty, where $U$'s onion is at this point in time. Let $S$ be this location (i.e., server). At the next round, all the onions that were routed to $S$ at the previous round, emerge from $S$ and are routed to their next destinations, each, from the adversary's perspective, having an equal $\frac{1}{k}$ probability of being $U$'s onion, where $k$ denotes the total number of onions that were routed to $S$ at round $1$. Continuing with this analysis, from the adversary's perspective, which onion is $U$'s onion becomes progressively more uncertain with every round of the protocol. Below, we show that after $\depth$ rounds, the adversary's ``belief'' of which onion is $U$'s onion is statistically indistinguishable from the uniform distribution over the $N$ possible onions. 

Fix a round $i$. 

Order the onions (at round $i$) from most likely to be $U$'s onion (from the adversary's point of view) to the least likely to be $U$'s onion (from the adversary's point of view); let $O_1, \dots, O_N$ denote this sequence of onions. 

W.l.o.g., assume that $3$ evenly divides $N$. 

Let $G_1 = \{O_1, O_2, \dots, O_{N/3-1}\}$ be the one-third most likely onions, 
let $G_3 = \{O_{2N/3}, O_{2N/3+1},\dots, O_N\}$ be the one-third \emph{least likely} onions, and 
let $G_2 = \{O_{N/3}, O_{N/3+1}\dots, O_{2N/3-1}\}$ be the remaining one-third of the onions ``in the middle''. 

Let $X_1$ and $X_2$ denote the likelihoods of the most likely onion in $G_1$ and the least likely onion in $G_1$, resp.; 
let $X_3$ and $X_4$ denote the likelihoods of the most likely onion in $G_2$ and the least likely onion in $G_2$, resp.; and
let $X_5$ and $X_6$ denote the likelihoods of the most likely onion in $G_3$ and the least likely onion in $G_3$, resp. 
Clearly, $X_1 \ge X_2 \ge X_3 \ge X_4 \ge X_5 \ge X_6$.  

Let $o_1$ be an onion with the maximal likelihood to be $U$'s at the next round (in round $i+1$), and let $x_1$ denote its likelihood. 

Let $B$ be the set of all onions at the next round that are routed to the same bin as $o_1$, and let $\ell$ denote the size of $B$, i.e., $\ell = |B|$. 

Let $\mathbb{E}$ denote the server load (i.e., the average number of onions per server per round). 
Then, the expected number of onions that each group $G_j$ contributes to $B$ is $\frac{\mathbb{E}}{3} = \frac{N}{3n} = \bigOmega{\log^{2} \secpar}$. 
From Chernoff bounds for Poisson trials (see Lemma~\ref{lem-mix}b in Appendix~\ref{sec:proofs}), with overwhelming probability, each group $G_j$ contributes to $B$ a number of onions that is arbitrarily close to this expected number. Thus, for every constant $d > 0$, 
\begin{align*}
x_1 = \prob{o_1} &\le \frac{(1/3 + d) \ell X_1 + (1/3) \ell X_3 + (1/3 - d) \ell X_5}{\ell} \\
&= \left(\frac{1}{3} + d\right) X_1 + \left(\frac{1}{3}\right) X_3 + \left(\frac{1}{3} - d\right) X_5
\intertext{\indent Let $x_6$ be the likelihood of the least likely onion $o_N$ in the next round. Following a similar argument as above, for every constant $d > 0$, }
x_6 = \prob{o_N} &\ge \frac{(1/3 - d) \ell X_2 + (1/3) \ell X_4 + (1/3 + d) \ell X_6}{\ell} \\
&= \left(\frac{1}{3} - d\right) X_2 + \left(\frac{1}{3}\right) X_4 + \left(\frac{1}{3} + d\right) X_6
\end{align*}

It follows that the gap $g$ between $x_i$ and $x_6$ can be bounded as follows: For every constant $d > 0$, 
\begin{align*}
g = x_1 - x_6 &\le \left(\left(\frac{1}{3} + d\right) X_1 + \left(\frac{1}{3}\right) X_3 + \left(\frac{1}{3} - d\right) X_5\right) \\
& \hspace{12mm} - \left(\left(\frac{1}{3} - d\right) X_2 + \left(\frac{1}{3}\right) X_4 + \left(\frac{1}{3} + d\right) X_6\right) \\
&= \frac{1}{3}(X_1 - X_2 + X_3 - X_4 + X_5 - X_6) + d(X_1 - X_6) + d(X_2 - X_5) \\
&\le \frac{1}{3}(X_1 - X_2 + X_3 - X_4 + X_5 - X_6) + 2d(X_1 - X_6) \\
&\le \left(\frac{1}{3} + 2d\right) (X_1 - X_6)
\intertext{In particular, letting $d = \frac{1}{12}$, }
g &\le \frac{1}{2}(X_1 - X_6)
\end{align*}

The gap between the most likely onion and the least likely onion is at least halved in every round. Thus, after a poly-logarithmic number of rounds, the gap is negligibly small.

In the proof above, the onions were partitioned into three groups at every round. By partitioning the onions into an appropriately large constant number of groups, we can show that $\pprotocol$ achieves statistical privacy after $\depth = \bigOmega{1}\log^{2} \secpar$ rounds.  
\end{proof}

We are now ready to prove the main result of this section: 
\begin{theorem} \label{lem:passive}
$\pprotocol$ is statistically private from the passive adversary capable of monitoring any constant fraction~$\corruptions\in[0, 1)$ of the servers when $\frac{N}{n} = \bigOmega{\log^{2} \secpar}$, and $\depth = \bigOmega{\log^{2} \secpar}$, where $\secpar \in \mathbb{N}$ denotes the security parameter. 
\end{theorem} 

\begin{proof}
We prove this by cases. 

In the first case, $\sigma_1$ is the same as $\sigma_0$ except that the inputs of two users are swapped, i.e., $d(\sigma_0, \sigma_1) = 2$. Using Chernoff bounds for Poisson trials (Lemma~\ref{lem-mix}b in Appendix~\ref{sec:proofs}), there are at least some polylog number of rounds where the swapped onions are both routed to an honest bin (not necessarily the same bin). From Theorem~\ref{thm-random}, after the polylog number of steps, the locations of these two target onions are statistically indistinguishable from each other. 

In the second case, $d(\sigma_0, \sigma_1) > 2$. However, the distance between $\sigma_0$ and $\sigma_1$ is always polynomially bounded. By a simple hybrid argument, it follows that $V^{\Pi, \adv}(\sigma_0) \approx_s V^{\Pi, \adv}(\sigma_1)$ from case~1.
\end{proof}

\textbf{Remark:}  Protocol~$\pprotocol$ is not secure from the active adversary. This is because, with non-negligible probability, any honest user will choose a corrupted party as its first hop on its onion's routing path, in which case the adversary can drop the target user's onion at the first hop and observe who does not receive an onion at the last round.

\section{The active adversary} \label{sec:pi3}
We now present an OR protocol, $\aprotocol$, that is secure from the active adversary. The setting for $\aprotocol$ is different from that of our previous solution in a couple of important ways. Whereas $\pprotocol$ is \emph{statistically} private from the passive adversary, $\aprotocol$ is only \emph{differentially} private from the active adversary. The upside is that we are no longer constrained to operate in the simple I/O setting; the input can be any valid input. 

We let $[N]$ be the set of $N$ parties participating in a protocol. Every party is both a user and a server. As before, $\mathcal{OR} = (\mathsf{Gen}, \mathsf{FormOnion}, \mathsf{ProcOnion})$ is a secure onion routing scheme; and for every $i\in[N]$, $(\pk_i, \sk_i) \gets \mathsf{Gen}(\secparam)$ denotes the key pair generated for party~$i$, where $\secpar$ is the security parameter. Further, we assume that every pair~$(i, k)\in[N]^2$ of parties shares a common secret key\footnote{In practice, the shared keys do not need to be set up in advance; they can be generated as needed from an existing PKI, e.g., using Diffie-Hellman.}, denoted by $\sk_{i, k}$. 

$F$ is a pseudorandom function (PRF). 

We describe the protocol by the setup and routing algorithms for party~$i\in[N]$; each honest party runs the same algorithms. 

\myparagraph{Setup} Let $L = \beta \log^2 \secpar$ for some constant $\beta >0$. 

During the setup phase, party~$i$ prepares a set of onions from its input. For every message pair~$u = \{m, j\}$ in party~$i$'s input, party~$i$ picks a sequence~$T_1^u, \dots, T_\depth^u$ of parties, where each party~$T_\ell^u$ is chosen independently and uniformly at random, and forms an onion from the message~$m$, the routing path~$(T_1^u, \dots, T_\depth^u, j)$, the public keys~$(\pk_{T_1^u}, \dots, \pk_{T_\depth^u}, \pk_{j})$, and a list of empty nonces.  

Additionally, party~$i$ forms some dummy onions, where a dummy onion is an onion formed using the empty message~$\bot$. 
\begin{enumerate}
\item \textbf{for} every index~$(r, k) \in [\depth] \times [N]$:
\begin{enumerate}
\item compute $b \gets F(\sk_{i, k}, \mathsf{session}+r, 0)$, where $\mathsf{session}\in\mathbb{N}$ denotes the protocol instance. 
\item \textbf{if} $b \equiv 1$ --- set to occur with frequency $\frac{\alpha \log^{2} \secpar}{N}$ for some constant $\alpha >0$ --- \textbf{do}:
\begin{enumerate}
\item choose a list~$T^{r, k} = (T_1^{r, k}, \dots, T_{r-1}^{r, k}, T_{r+1}^{r, k}, \dots,T_{\depth+1}^{r, k})$ of parties, where each party is chosen independently and uniformly at random; 
\item create a list~$s^{r, k} = (s_1^{r, k}, \dots, s_\depth^{r, k})$ of nonces, where 
\[
s_r^{r, k} = (\mathsf{checkpt}, F(\sk_{i, k}, \mathsf{session}+r, 1)) ,
\]
and all other elements of $s^{r, k}$ are $\bot$; and
\item form a dummy onion using the message~$\bot$, the routing path~$T^{r, k}$, the public keys associated with $T^{r, k}$, and the list~$s^{r, k}$ of nonces. 
\end{enumerate}
\item \textbf{end if} 
\end{enumerate}
\item \textbf{end for}
\end{enumerate}

The additional information~$s_r^{r, k}$ is embedded in only the $r$-th layer; no additional information is embedded in any other layer. At the first round of the protocol run, all formed onions are sent to their respective first hops. 

\myparagraph{Routing} 
If party~$i$ forms a dummy onion with nonce~$s_r^{r, k}$ embedded in the $r$-th layer, then it expects to receive a symmetric dummy onion at the $r$-th round formed by another party~$k$ that, when processed, reveals the same nonce~$s_r^{r, k}$. If many checkpoint nonces are missing, then party~$i$ knows to abort the protocol run. 

After every round~$r\in[\depth]$ (but before round~$r+1)$, party~$i$ peels the onions it received at round~$r$ and counts the number of missing checkpoint nonces. If the count exceeds a threshold value~$t$, the party aborts the protocol run; otherwise,  at round~$r+1$, the peeled onions are sent to their next destinations in random order. After the final round, party~$i$ outputs the set of messages revealed from processing its the onions it receives at round~$\depth+1$. 

\myparagraph{Correctness and efficiency} Recalling that correctness is defined with respect to the passive adversary, $\aprotocol$ is clearly correct. Moreover, unless an honest party aborts the protocol run, all messages that are not dropped by the adversary are delivered to their final destinations. In $\aprotocol$, the communication complexity blow-up is $O(\log^{6} \secpar)$, since the latency is $\depth+1 = O(\log^2 \secpar)$ rounds, and the server load is $O(\log^{4} \secpar)$. 

\myparagraph{Privacy} 
To prove that $\aprotocol$ is secure, we require that the thresholding mechanism does its job: 
\begin{lemma} \label{lem:main} 
In~$\aprotocol$, if $F$ is a random function, 
$
t = c(1-d)(1-\corruptions)^2 \alpha \log^{2} \secpar 
$
for some $c, d \in (0, 1)$, and an honest party does not abort within the first $r$~rounds of the protocol run, then with overwhelming probability, at least $(1 - c)$ of the dummy onions created between honest parties survive at least $(r-1)$ rounds, even in the presence of an active adversary non-adaptively corrupting a constant fraction $\corruptions \in [0, 1)$ of the parties. 
\end{lemma}

The proof relies on a known concentration bound for the hypergeometric distribution~\cite{HS05} and can be found in Appendix~\ref{sec:proofs}. 

\begin{theorem} \label{thm:main}
If, in $\aprotocol$, $F$ is a random function, $N \ge \frac{3}{1-\corruptions}$, and 
$
t = c(1-d)(1-\corruptions)^2 \alpha \log^{2} \secpar 
$
for some $c, d \in (0, 1)$, then, for 
$
\alpha\beta \ge -\frac{36 (1+\epsilon/2)^2 \ln \left(\delta/4\right)}{(1-c) (1-\corruptions)^2 \epsilon^2} ,
$
$\aprotocol$ is $(\epsilon, \delta)$-DP from the active adversary non-adaptively corrupting a constant fraction $\corruptions \in [0, 1)$ of the parties. 
\end{theorem}

\begin{proof} The proof is by cases. 

\myparagraph{Case 1} All honest parties abort within the first half of the protocol run. 
In this case, with overwhelming probability, no onion created by an honest party will be delivered to its final destination. This is because each intermediary hop of every onion created by an honest party is chosen independently and uniformly at random. Thus, there is only a negligible probability of the latter half of an onion's routing path passing through only corrupted parties (Lemma~\ref{lem-mix}b). Since there are a polynomial number of onions, by the union bound, the probability that there exists an onion whose routing path passes through only corrupted parties is negligibly small. Since no onion created by an honest party is delivered to its final destination, the adversary doesn't learn anything.  

\myparagraph{Case 2} Some honest party doesn't abort within the first half of the protocol run. Let $\adv$ be any adversary that non-adaptively corrupts a constant $\corruptions \in [0, 1)$ of the parties. Suppose that for every onion that survive the first half of the protocol run, a dark angel provides the adversary~$\adv$ with the second half of the onion's routing path. Further suppose that no other onions are dropped in the second half of the protocol run. (If more onions are dropped, then $\aprotocol$ is secure from the post-processing theorem for differential privacy~\cite[Proposition~2.1]{Dworkbook}.) 

For any two neighboring inputs $\sigma_0$ and $\sigma_1$, the only difference in the adversary's views, $V^{\aprotocol, \adv}(\sigma_0)$ and $V^{\aprotocol, \adv}(\sigma_1)$, is the routing of a single onion~$O$. If there is an honest party who does not abort within the first half of the protocol run, then from Lemma~\ref{lem:main}, some constant fraction of the dummy onions created by the honest parties survive the first half of the protocol run with overwhelming probability. So, from Theorem~\ref{lem:passive}, the onions are no longer linked to their senders by the end of the first half of the protocol run. Thus, the only information that $\adv$ could find useful is the volume of onions sent out by the sender~$P_s$ of the extra onion~$O$ and the volume of onions received by the receiver~$P_r$ of $O$. 

Let $X$ denote the number of dummy onions created by $P_s$. For every $(k, r) \in [\depth]\times[N]$, an honest sender~$P_s$ creates a dummy onion with probability $\frac{\alpha \log^{2} \secpar}{N}$; so $X \sim \mathsf{Binomial}({H}, {p})$, where ${H} = LN$, and ${p} = \frac{\alpha \log^{2}\secpar}{N}$. 

Let $Y \sim \mathsf{Binomial}({G}, {q})$ be another binomial random variable with parameters~${G} = \frac{\depth(1-\corruptions)^2 N^2}{3}$ and ${q} = \frac{(1-c) \alpha \log^{2}\secpar}{N^2}$. For $N \ge \frac{3}{1-\corruptions}$ and sufficiently small $d>0$, ${G} \le (1-d) L \binom{(1-\corruptions)N -1}{2}$; thus, with overwhelming probability, $Y$ is less than the number of dummy onions created between honest non-$P_s$ parties and received by $P_r$ in the final round (Chernoff bounds).

Let $\mathcal{O} \stackrel{\text{def}}{=} \mathbb{N}\times\mathbb{N}$ be the sample space for the multivariate random variable~$(X, Y)$. 

Let $\mathcal{O}_1$ be the event that $|X - \mathbb{E}[X]| \le d'\mathbb{E}[X]$, and $|Y - \mathbb{E}[Y]| \le d'\mathbb{E}[Y]$, where $d' = \frac{\epsilon/2}{1+\epsilon/2}$, $\mathbb{E}[X] = {H}{p}$ is the expected value of $X$, and $\mathbb{E}[Y] = {G}{q}$ is the expected value of $Y$; and let $\bar{\mathcal{O}_1}$ be the complement of $\mathcal{O}_1$. 

For every $(x, y) \in \mathcal{O}_1$, we can show that
\begin{align}
\max\left(\frac{\prob{(X, Y) = (x, y)}}{\prob{(X, Y) = (x+1, y+1)} }, \frac{\prob{(X, Y) = (x+1, y+1)}}{\prob{(X, Y) = (x, y)}} \right) \le e^\epsilon .\label{eq:footnote1} 
\end{align}
We can also show that the probability of the tail event $\bar{\mathcal{O}_1}$ occurring is negligible in $\secpar$ and at most $\delta$ when 
$
\alpha\beta \ge -\frac{36 (1+\epsilon/2)^2 \ln \left(\delta/4\right)}{(1-c) (1-\corruptions)^2 \epsilon^2} . \label{eq:footnote2}
$
(See Appendix~\ref{sec:proofs}.) 

Any event $\mathcal{E}$ can be decomposed into two subsets $\mathcal{E}_{1}$ and $\mathcal{E}_{2}$, such that (1)~$\mathcal{E} = \mathcal{E}_{1} \cup \mathcal{E}_{2}$, (2)~$ \mathcal{E}_{1} \subseteq \mathcal{O}_1$, and (3)~$\mathcal{E}_{2} \subseteq \bar{\mathcal{O}_1}$. It follows that, for every event $\mathcal{E}$,
\begin{align}
\prob{(X,Y) \in \mathcal{E}} &\le e^{\epsilon} \cdot \prob{(X+1, Y+1) \in \mathcal{E}} + \delta  \text{, and} \label{eq:first} \\
\prob{(X+1,Y+1) \in \mathcal{E}} &\le e^{\epsilon} \cdot \prob{(X,Y) \in \mathcal{E}} + \delta . \label{eq:second}
\end{align}

The views $V^{\aprotocol, \adv}(\sigma_0)$ and $V^{\aprotocol, \adv}(\sigma_1)$ are the same except that $O$ exists in one of the views but not in the other. Thus, \eqref{eq:first} and \eqref{eq:second} suffice to show that for any set $\mathcal{V}$ of views and for any $b \in \{0, 1\}$, 
$
\prob{V_b^{\aprotocol, \adv} \in \mathcal{V}}  \le e^{\epsilon} \cdot \prob{V_{\bar{b}}^{\aprotocol, \adv} \in \mathcal{V}} + \delta ,
$
where $\bar{b} = b+1 \mod 2$. 
\end{proof}

\textbf{Remark:} 
Our protocols are for a \emph{single-pass} setting, where the users send out messages once. It is clear how our statistical privacy results would compose for the multi-pass case. To prove that $\aprotocol$ also provides differential privacy in the multi-pass scenario --- albeit for degraded security parameters --- we can use the $k$-fold composition theorem~\cite{Dworkbook}; the noise falls at a rate of the square-root of the number of runs. 

\subsection*{Acknowledgements}
The work was supported in part by NSF grants IIS-1247581 and CNS-1422361. 
We thank Dov Gordon for pointing out an issue in the proof of Theorem~\ref{thm-random} in a prior version of this paper; the proof has been fixed for this version!

\bibliographystyle{is-alpha}
\bibliography{OR-protocols}


\appendix 

\section{In the standard model} \label{sec:reallife} 

A protocol that uses cryptographic tools like public-key encryption and authentication can only be differentially private when idealized versions of the tools are used. In practice, the best we can achieve is the computational analogue.
\begin{definition} [$\epsilon$-IND-CDP~\cite{C:MPRV09}]
Let $\Sigma_\secpar$ be the set of all possible polynomially sized (w.r.t.~$\secpar$) input vectors. A communications protocol $\Pi$ is \emph{$\epsilon$-IND-CDP} from every adversary in a class $\{\adv_\secpar\}_{\secpar \in \mathbb{N}}$ of non-uniform probabilistic polynomial-time (PPT) adversaries if for any adversary~$\adv_\secpar$, any two neighboring inputs $\sigma_{0, \secpar}, \sigma_{1, \secpar} \in \Sigma_{\secpar}$ that differ only on the honest parties' inputs, and any polynomially sized advice string $z_\secpar$,
\begin{align*}
&\prob{\adv_\secpar(1^\secpar, V^{\Pi, \adv}_{0, \secpar}, z_\secpar) = 1} \le e^\epsilon \cdot \prob{\adv_\secpar(1^\secpar, V^{\Pi, \adv}_{1, \secpar}, z_\secpar) = 1} + \negl,
\end{align*}
where $\secpar \in \mathbb{N}$ denotes the security parameter. 
\end{definition}

In Theorem~\ref{thm:main} we proved that $\aprotocol$ achieves differential privacy when the onions were constructed using idealized (information-theoretically secure) public-key encryption, and when the function $F$ for generating checkpoint onions is truly random. Such idealized encryption does not exist, and truly random functions are not available, either.  So our actual proposed solution is to use standard computationally secure tools instead. In this case, $\aprotocol$ achieves $\epsilon$-IND-CDP: 

\begin{theorem} \label{thm:main2}
If $\aprotocol$ is constructed using a CCA2-secure encryption scheme and pseudorandom function $F$, and if in $\aprotocol$,
$n \ge \frac{3}{1-\corruptions}$, and 
\[
t = c(1-d)(1-\corruptions)^2 \alpha \log^{2} \secpar
\]
for some $c, d \in (0, 1)$, then for sufficiently large protocol parameters $\alpha$ and $\beta$, $\aprotocol$ is $\epsilon$-IND-CDP from an active adversary capable of (non-adaptively) corrupting a constant fraction $\corruptions \in [0, 1)$ of the parties.
\end{theorem}

\begin{proof}
The proof is by a relatively standard hybrid argument; here we present it in the language of Canetti's universal composability~\cite{FOCS:Canetti01}. Let $\mathcal{F}_{\mathsf{enc}}$ be an ideal public-key encryption functionality and $\mathcal{F}_{\mathsf{PRF}}$ be an ideal pseudorandom function functionality.  

Recall that a CCA2-secure cryptosystem UC-realizes the (non-adaptive) $\mathcal{F}_{\mathsf{enc}}$ functionality, while a secure PRF UC-realizes the $\mathcal{F}_{\mathsf{PRF}}$ functionality.  That means that (1) there exist simulators $S_{\mathsf{enc}}$ and $S_{\mathsf{PRF}}$ that, in combination with $\mathcal{F}_{\mathsf{enc}}$ and $\mathcal{F}_{\mathsf{PRF}}$, respectively, realize idealized non-adaptive public-key encryption and a pseudorandom function, respectively.   Moreover, (2) for any PPT environment and adversary, the view that the adversary obtains in the $(\mathcal{F}_{\mathsf{enc}}, \mathcal{F}_{\mathsf{PRF}})$-hybrid model with these simulators is indistinguishable from its view when the CCA2 scheme and the PRF are used instead.

Let $\{\sigma_{0,\lambda}\}$ and $\{\sigma_{1,\lambda}\}$ be a sequence of neighboring inputs to $\aprotocol$, parameterized by the security parameter $\lambda$, and let $\adv$ be a non-uniform PPT adversary.  By $V_{\mathsf{ideal}}^{\adv, \Pi}(\sigma_{b,\lambda})$ let us denote $\adv$'s view in the protocol $\Pi$ with input $\sigma_{b,\lambda}$ where the encryption and PRF are realized using their ideal functionalities with the simulators.  By  $V_{\mathsf{real}}^{\adv, \Pi}(\sigma_{b,\lambda})$, let us denote the corresponding view where they are realized using the CCA2 encryption scheme and the PRF.

Using (1), by Theorem~\ref{thm:main}, $\aprotocol$ provides $(\epsilon,\delta)$-DP privacy when the encryption and the PRF are realized using the simulator with $\mathcal{F}_{\mathsf{enc}}$ and $\mathcal{F}_{\mathsf{PRF}}$, and so 
\begin{align*}
&\prob{\adv_\lambda(1^\lambda, V_\mathsf{ideal}^{\adv, \aprotocol}(\sigma_{0,\lambda}), z_\lambda) = 1} \le e^\epsilon \cdot \prob{\adv_\lambda(1^\lambda, V_\mathsf{ideal}^{\adv, \aprotocol}(\sigma_{1,\lambda}), z_\lambda) = 1} + \delta .
\end{align*}
Using (2), we know that for $b \in \{0,1\}$,
\begin{align*}
&|\prob{\adv_\lambda(1^\lambda, V_\mathsf{real}^{\adv, \aprotocol}(\sigma_{b,\lambda}), z_\lambda) = 1} -  \prob{\adv_\lambda(1^\lambda, V_\mathsf{ideal}^{\adv, \aprotocol}(\sigma_{b,\lambda}), z_\lambda) = 1}| \leq \negl.
\end{align*}
Putting the two together, and keeping in mind that $\delta$ is negligible in $\lambda$, we get the theorem.

\end{proof}

\myparagraph{Deriving pairwise PRF keys using a PKI}
Recall that $\aprotocol$ requires that any pair of clients, $P_i$ and $P_j$, are in possession of a shared PRF key $sk_{i,j}$.  What would be more desirable is if they could derive a shared PRF key using the public-key infrastructure.

Under the decisional Diffie-Hellman assumption, in the non-adaptive corruption setting this can be accomplished as follows. Let $g$ be a generator of a group $G$ of order $q$ in which the decisional Diffie-Hellman (DDH) assumption holds; and suppose that these are publicly available.  Let a random $x_i\in Z_q$ be part of every $P_i$'s secret key, and let $Y_i = g^{x_i}$ be part of $P_i$'s public key.  Then, $Z_{i,j} = Y_i^{x_j} = Y_j^{x_i}$ can serve as a PRF key for an appropriate choice of a PRF $F'$, namely, one that is a PRF when the seed is chosen as a random element of the group $G$. For example, using the leftover hash lemma~\cite{STOC:ImpLevLub89}, we can use $F'_{Z_{i,j}}(a) = F^*_{H(Z_{i,j})}(a)$ where $H$ is a universal hash function, and $F^*$ is any PRF.  

To sum it up, our proposed construction for how $P_i$ computes the shared PRF for $P_j$ on input $a$ is: $F(sk_i,pk_j,a) = F'_{Y_j^{x_i}}(a)$, where $x_i\in Z_q$ is part of $P_i$'s secret key, $Y_j = g^{x_j}$ is part of $P_j$'s public key, and $F'$ is any PRF keyed by random elements of $G$.

To show that no PPT adversary can distinguish whether (Case A) all honest $P_i$ and $P_j$ are using this construction or (Case B) a truly random function, we will use a standard hybrid argument.

Consider the hybrid experiment $E_H$ in which all honest pairs $P_i$ and $P_j$ have a truly random $s_{i,j} \in G$, and use it as a seed to $F'$.  It is straightforward to see that an adversary distinguishing this experiment from Case B breaks the PRF $F'$.

Suppose that we have $N'$ honest clients.  So now consider a series of experiments $E_0, \dots, E_{\binom{N'}{2}}$ in which the honest pairs of clients are ordered in some fashion, and the first $u$ pairs will use truly random seeds $s_{i,j}$ while the rest compute $s_{i,j} = Y_j^{x_i} = Y_i^{x_j}$.  Since $E_0$ is identical to $E_H$ above, and $E_{\binom{N'}{2}}$ is identical to Case A, by the hybrid argument it is sufficient to show that, for every $u$, no adversary can distinguish $E_u$ from $E_{u+1}$.  

Our reduction is given a decisional Diffie-Hellman challenge $(g,A,B,C)=(g,g^\alpha,g^\beta,g^\gamma)$ and needs to distinguish whether $\gamma = \alpha\beta$ or random.  Let $P_I$ and $P_J$ be the $(u+1)^{\mathsf{st}}$ pair of honest clients.  Let $Y_I = A$ and $Y_J = B$, and for all the other $P_i$, let $Y_i = g^{x_i}$ for some $x_i$ that the reduction knows. For the first $u$ pairs of honest clients, $P_i$ and $P_j$, let $s_{i,j}$ be a truly random PRF seed picked by the reduction. When the adversary queries for $F(sk_i,pk_j,a)$, the reduction responds with $F'_{s_{i,j}}(a)$ whenever $(P_i,P_j)$ are among the first $u$ pairs, with $F'_{Z_{i,j}}(a)$ if they are among the last ${\binom{N'}{2}}-u-1$ pairs, where $Z_{i,j} = Y_i^{x_j}$ and the reduction knows, w.l.o.g.\, $x_j$, and with $F'_C(a)$ if $i = I$ and $j = J$.  If $\gamma = \alpha\beta$, then the view the adversary gets is identical to $E_u$; but if it is random, then it is identical to $E_{u+1}$.  This completes the proof.

\section{Supplementary proofs} \label{sec:proofs}
\subsection{A useful lemma from Chernoff bounds}
We make use of the following facts derived from Chernoff bounds, which allow us to make the arguments that certain favorable events occur with overwhelming probability as opposed to merely occurring in expectation. 

\begin{lemma} \label{lem-mix} If $N = \mathsf{poly}(\secpar)$ balls are thrown independently and uniformly at random into $n = O(N/\log^{2} \secpar)$ bins, then
\begin{itemize} 
\item[a.] For any $0 < d \le 1$, the number of balls thrown into any bin is at least $(1-d) \frac{N}{n}$ and at most $(1+d) \frac{N}{n}$ with overwhelming probability with respect to the parameter~$\secpar$. 
\item[b.] For any $0 < d \le 1$, the number of balls thrown into any set of $k\in[n]$ bins is at least $(1-d) \frac{kN}{n}$ and at most $(1+d) \frac{kN}{n}$ with overwhelming probability with respect to the parameter~$\secpar$. 
\end{itemize}
\end{lemma}

\begin{proof}
{(of a.)} For all $i \in [n]$, let $O_i$ denote the number of balls in bin $i$. The probability that the number of balls in a fixed bin $i \in [n]$ is significantly off from the expected number of balls can be bounded by a Chernoff bound~\cite[Cor.~4.6]{MU05}:  
\[
\prob{ |O_i - \mathbb{E}[O_i]| \ge d\mathbb{E}[O_i] } \le  2e^{- d^2 \mathbb{E}[O_i]/3 } .
\]
So, by the union bound, the probability that any bin has significantly more or less than the expected number of balls is negligible in $\secpar$: 
\begin{align*}
&\prob{exists i:\, |O_i - \mathbb{E}[O_i]| \ge d\mathbb{E}[O_i] } \le \frac{2n}{e^{d^2 \cdot \Theta(\log^2) \secpar/3}} = \frac{1}{\secpar^{\Theta(\log^c \secpar)}} . 
\end{align*}

{(of b.)} 
We prove this by contradiction. We assume the (absolute value) difference between the actual number of balls in $k$ bins and the expected number of balls in $k$ bins can exceed $d k {E}[O_i]$ with non-negligible probability. By the pigeonhole principle, there must exist a bin such that the (absolute value) difference between the actual number of balls its holds and the expected number of balls for a single bin exceeds $d \mathbb{E}[O_i]$. This contradicts part~a of the lemma. 
\end{proof}

\subsection{Proof of Lemma~\ref{lem:main}} 
\begin{proof}
Let $\mathcal{A}$ be any active adversary capable of non-adaptively corrupting a constant fraction $\corruptions \in [0, 1)$ of the network. Since $\mathcal{A}$ controls the corrupted parties, it can know the checkpoint round, location, and nonce of any onion created between a corrupted party and any other party. Thus, we assume that any onion created between a corrupted party and any other party can be replaced by $\mathcal{A}$ without the replacement being detected by any honest party. Suppose that $\mathcal{A}$ has help from a dark angel who marks every onion created between a corrupted party and any other party, so that $\mathcal{A}$ can replace all marked onions without detection. Even so, without eliminating some unmarked onions, some positive constant fraction of the dummy onions would survive (Lemma~\ref{lem-mix}b). 

Let an onion created between two honest parties be called \emph{unmarked}, and consider only unmarked onions. For any onion with a checkpoint in the future, the probability that the adversary $\mathcal{A}$ can drop the onion without any honest party detecting that the onion was dropped is negligibly small; $\mathcal{A}$ cannot produce the correct checkpoint nonce with sufficiently high probability. 

At any round $r$, $\mathcal{A}$ is unable to distinguish between any two unmarked onions. Let $u$ denote the total number of unmarked onions. Again relying on Chernoff bounds (Lemma~\ref{lem-mix}b), 
\[
u \ge (1-d_1) (1-\corruptions)^2 \alpha  L N \log^{2} \secpar 
\]
for any $0 < d_1 < d$.  

Let $v_r$ denote the cumulative number of unmarked onions that have been eliminated by $\mathcal{A}$ so far at round $r$. If an honest party~$i$ does not detect more than 
\[
\frac{(1-d_2)cu}{LN} \ge (1-d_1)(1-d_2)c (1-\corruptions)^2 \alpha \log^{2} \secpar = (1-d)c (1-\corruptions)^2 \alpha \log^{2} \secpar
\]
missing onions, then with overwhelming probability, $v_{r-1} \le c u$. (This follows from a known concentration bound for the hypergeometric distribution~\cite{HS05}). Thus, at least $1 - c$ of all dummy onions created between honest parties survive until round $r-1$, with overwhelming probability. 
\end{proof}

\subsection{Proof of Equation~(2)}
\begin{proof} 
Let $E$ be the event that $|Y - \mathbb{E}[Y]| \le d'\mathbb{E}[Y]$, and let $f(\cdot)$ denote the probability mass function of $Y$. Letting $y' \stackrel{\text{def}}{=} \argmax_{y\in E} \frac{\max(f(y), f(y+1))}{\min(f(y), f(y+1))}$, 
\begin{align*}
y' = (1-d'){G}{q} = \frac{{G}{q}}{1+\epsilon}  = \frac{C_GC_q}{1+\epsilon} \log^{2} \secpar 
\end{align*}
where  $C_G = \frac{L(1-\corruptions)^2}{3}$, and $C_q = (1-c)\alpha$. This is true, because 
\[
\max_{y: {G}{q}-d \le y \le {G}{q}} \frac{f(y+1)}{f(y)} > \max_{y: {G}{q} \le y \le {G}{q}+d} \frac{f(y)}{f(y+1)} 
\]
whenever ${q} < \frac{1}{2}$. 
\begin{align*}
\frac{f(y'+1)}{f(y')} &= \frac{\binom{{G}}{y'+1}{q}^{y'+1}(1-{q})^{N-y'-1}}{\binom{{G}}{y'}{q}^{y'}(1-{q})^{N-y'}} \\
&= \frac{{G}! {q}^{y'+1} (1-{q})^{N-y'-1}}{(y'+1)!({G}-y'-1)!} \cdot \frac{y'!({G}-y')!}{H! {q}^{y'} (1-{q})^{N-y'}} \\
&= \frac{({G}-y'){q}}{(y'+1)(1-{q})} \\
&= \frac{\left(C_G N^2 - \frac{C_GC_q}{1+\epsilon} \log^{2} \secpar\right) C_q \frac{\log^{2} \secpar}{N^2}}{\left(\frac{C_GC_q}{1+\epsilon} \log^{2} \secpar + 1\right)\left(1-C_q \frac{\log^{2} \secpar}{N^2}\right)} \\
&= \frac{1+\epsilon/2}{1+\epsilon/2} \frac{\left(C_G N^2 - \frac{C_GC_q}{1+\epsilon} \log^{2} \secpar\right)}{\left(\frac{C_GC_q}{1+\epsilon} \log^{2} \secpar + 1\right)} \cdot 
\frac{N^2}{N^2} \frac{C_q \frac{\log^{2} \secpar}{N^2}}{\left(1-C_q \frac{\log^{2} \secpar}{N^2}\right)} \\
&= \frac{C_q \log^{2} \secpar}{C_GC_q\log^{2} \secpar + (1+ \epsilon/2)} \cdot \frac{(1+\epsilon/2)C_G N^2 - C_GC_q \log^{2} \secpar}{N^2 - C_q \log^{2} \secpar} \\
&= \frac{C_GC_q \log^{2} \secpar}{C_GC_q\log^{2} \secpar + (1+ \epsilon/2)} \cdot \frac{(1+\epsilon/2) N^2 - C_q \log^{2} \secpar}{N^2 - C_q \log^{2} \secpar} \\
&\le \frac{(1+\epsilon/2) N^2 - C_q \log^{2} \secpar}{N^2 - C_q \log^{2} \secpar} \le 1 + \epsilon/2 .
\end{align*}

Since $1 + \epsilon/2 \le e^{\epsilon/2}$ whenever $\epsilon/2 \ge 0$,
\[
\max\left(\frac{\prob{Y = y}}{\prob{Y = y+1}}, \frac{\prob{Y = y+1}}{\prob{Y = y}} \right) \le e^{\epsilon/2} 
\]
for every outcome $y \in E$. Using a similar argument, we can show that 
\[
\max\left(\frac{\prob{X = x}}{\prob{X = x+1}}, \frac{\prob{X = x+1}}{\prob{X = x}} \right) \le e^{\epsilon/2} 
\]
for every outcome $x \in [{H}{p} - d', {H}{p} + d']$. Since $X$ and $Y$ are independent, this completes our argument. 
\end{proof}

\subsection{Proof of Equation~(3)}
\begin{proof}
From \cite[Cors.~4.6]{MU05}, 
\[
\prob{|Y - \mathbb{E}[Y]| \ge d'\mathbb{E}[Y]} \le 2e^{-\mathbb{E}[Y] (d')^2/3} ,
\]
and $2e^{-\mathbb{E}[Y] (d')^2/3} \le \delta/2$ when  
\[
\alpha\beta \ge -\frac{36 (1+\epsilon/2)^2 \ln \left(\delta/4\right)}{(1-c) (1-\corruptions)^2 \epsilon^2} . 
\]
The tail event for $X$ can be bounded in a similar fashion. 
\end{proof}

\section{OR protocols for the network adversary} \label{sec:pi1}
In this section, each user sends a single message and receives a single message. Stated more precisely, the input is of the form 
\[
\sigma = (\{ (m_1, \pi(1))\}, \dots, \{ (m_N, \pi(N))\}),
\]
where $\pi : [N] \rightarrow [N]$ is any permutation function over the set $[N]$, and $m_1, \dots, m_N$ are any messages from the message space $\mathcal{M}$. 

\subsection{The basic construction}
Our first construction, $\nprotocol$, builds on Ohrimenko et al.'s basic oblivious permutation algorithm \cite{ICALP:OGTU14}. A node is either a user client or a server that serves as a mix-node. We let $[N]$ be the set of users, and $\mathcal{S} = \{S_1, \dots, S_{n}\}$ the set of servers, where $N = n^2$. The scheme~$\mathcal{OR} = (\mathsf{Gen}, \mathsf{FormOnion}, \mathsf{ProcOnion})$ is a secure onion routing scheme; and for every $i\in[N]$, $(\pk_i, \sk_i) \leftarrow \mathsf{Gen}(1^\secpar)$ denotes the key pair generated for party~$i$. 

Every user creates an onion during setup, and everyone's onions are released simultaneously at the start of the protocol run. Every user $i \in [N]$ constructs an onion with a routing path of length three, where the first two hops (the intermediary nodes)~$T_1$ and $T_2$ are chosen independently and uniformly at random from the set $\mathcal{S}$ of servers. Given the input~$\{(m, j)\}$, the onion is formed from the message~$m$, the routing path~$(T_1, T_2, j)$, $(\pk_{T_1}, \pk_{T_2}, \pk_j)$, and the list $(\bot, \bot)$ of empty nonces. 

All onions are released simultaneously at the first round of the protocol run. 
In between the first round and the second round, each first-hop server $S_i$ peels off the outermost layer of its received onions and creates a packet of size $k = \alpha \log \secpar$ onions for each second-hop server $S_j$, where $\secpar$ is the security parameter; each packet is a set of all onions for $S_j$ plus dummy onions as required\footnote{If the server needs to send more than $\alpha \log \secpar$ onions to some recipient, the packet size will be higher and without any additional dummy onions.}. The prepared onions, including dummies, are sent to their next destinations at the second round. Between the second and third rounds, every second-hop server $S_j$ peels off the outermost layers of its received onions, drops the dummy onions, and sends the real onions to their intended recipients in random order. Upon receiving an onion, each user processes its received onion and outputs the revealed message. (In Ohrimenko et al.'s construction, ``each user'' is directly connected to an assigned ``server node''; whereas in our construction, each intermediary hops on a routing path is chosen independently and uniformly at random. This allows us to cut the latency by half.)

Clearly, $\nprotocol$ is correct. The communication complexity blow-up of $\nprotocol$ is $O(\log \secpar)$; $N$ onions are transmitted from the users to the servers in round 1, $\alpha N \log \secpar$ onions are transmitted from servers to servers in round 2 with overwhelming probability, and $N$ onions are transmitted from servers back to user in round 3. The latency of $\nprotocol$ is three rounds, and the server load is $O(\sqrt{N} \log \secpar)$. 

\begin{theorem}\label{thm:netadv} 
Protocol $\nprotocol$ is statistically private against the network adversary for $\alpha \ge \frac{e}{2}$.
\end{theorem}

\begin{proof}
We can bound the probability of the event of a packet overflow as follows:
For any chosen pair of indices $(i, j) \in [n]^2$, the number~$X_{i, j}$ of onions that are routed first to $S_i$ and then to $S_j$ can be modeled by a Poisson random variable with event rate $1$. This is true because $N = n^2$ onions are each binned to an independently and uniformly random index-pair $(i, j) \in [n]^2$, and so we may approximate the numbers of onions that map to the $n^2$ pairs of indices as independent Poisson random variables~\cite[Ch.~5.4]{MU05}. 

From above, the probability that at least $k$ onions are first routed to $S_i$ and then to $S_j$ is equivalent to the probability that a realization of a random variable $X_{i, j} \sim {\bf Poisson}(1)$ is at least $k$. We can bound the probability of this event using Chernoff bounds for Poisson random variables~\cite[Thm.~4.4.1]{MU05}. Since $k > 1$ for sufficiently large $\secpar$: 
\[
\prob{ X_{i, j} \ge k } \le \frac{1}{e} \left( \frac{e}{k} \right)^k .
\]
So, from the union bound and the hypothesis that $\alpha \ge \frac{e}{2}$, 
\[
\prob{ \exists (i, j)  \in [n]^2 : X_{i, j} \ge k } \le \frac{n^2}{e} \left( \frac{e}{2\alpha \log \secpar} \right)^{2\alpha \log n} = \frac{1}{\secpar^{\Omega(\log \log \secpar)}} = \negl .
\]

Since the probability of an overflow is negligibly small and independent of the input vector, we can conclude that $\nprotocol$ is statistically private.
\end{proof}

\subsection{The tunable construction}
Our tunable OR protocol, $\nprotocol^+$, is an adaptation of Ohrimenko et al.'s optimized oblivious permutation algorithm \cite{ICALP:OGTU14}. 

The protocol description (below) refers to a butterfly network with branching factor~$B$ and height~$H$, where the number $n' = B^{H-1}$ of processor nodes is set to 
$
n' = \frac{n}{\alpha \log^2 \secpar} 
$
for some positive constant $\alpha > 0$. The processor nodes also serve as the switching nodes at every depth~$i\in[H]$ of the butterfly network. We define a set of nodes be a \emph{subnet} if every edge with an endpoint in the set has as the other endpoint, a node in the set. 

In $\nprotocol^+$, there are $N$ users, $1, \dots, N$; and there are $n'$ servers, $T_1, \dots, T_{n'}$, corresponding to the processor nodes. Every user creates an onion during setup, and everyone's onions are released simultaneously at the start of the protocol run. To create an onion, user~$i$ first chooses random entry and exit points $W_1, W_H \xleftarrow{\$} [n']$. The routing path of the onion is set to the path $p  = W_1, \dots, W_{H}$ from the entry $W_1$ to the exit $W_H$ in the butterfly network. On input~$\sigma_i = \{(m, j)\}$, user~$i$ forms an onion from the message~$m$, the recipient~$j$, the routing path $p$, the public keys of the parties associated with the parties on $p$, and a list of empty nonces.  At the first round of the protocol run, the user sends the formed onion to the first hop~$W_1$. 

For every $r\in[H]$, in between rounds~$r$ and $r+1$ of the protocol run, each server~$T_i$ peels off the outermost layer of each received onion. Let $\secpar$ denote the security parameter. For every next destination~$T_j$, $T_i$ creates a packet of $k = (1+d) \alpha \log^2 \secpar$ onions, containing all the onions whose next destination is $T_j$, plus dummy onions as needed. These formed onions are sent to their next destinations at round~$r+1$. In the final round $r = H+1$, the servers deliver the onions to the users. Upon receiving an onion, each user processes it and outputs the revealed message. (In Ohrimenko et al.'s construction, ``each user'' is directly connected to an assigned ``server node''; whereas in our construction, each intermediary hops on a routing path is chosen independently and uniformly at random. This allows us to cut the latency by half.) 

Clearly, $\nprotocol^+$ is correct. 

\begin{theorem} \label{thm:butterfly}
Letting $\secpar$ denote the security parameter, for any $B \in [\frac{\sqrt{N}}{\log^2 \secpar}]$, protocol~$\nprotocol^+$ can be set to run in $O(\frac{\log N}{\log B})$~rounds with server load~$O(B\log^2 \secpar)$ and communication complexity overhead~$O(\frac{B\log N \log^2 \secpar}{\log B})$. 
\end{theorem}

\begin{proof}
The proof follows from setting $B$ to branching factor in the butterfly network. 
\end{proof}

\begin{theorem} 
Protocol $\nprotocol^+$, with security parameter $\secpar \in \mathbb{N}$, is statistically private against the network adversary for $n = \mathsf{poly}(\secpar)$.
\end{theorem}

\begin{proof}
Conditioned on no server-to-server packet exceeding size $k = (1+d) \alpha \log^2 \secpar$, the adversary's view is the same regardless of the input. Moreover, the probability of a packet overflow is negligibly small (Lemma~\ref{lem-mix}).
\end{proof}

\end{document}